\documentclass[pra,twocolumn,floatfix,notitlepage,groupedaddress,longbibliography]{revtex4-1}

\usepackage[hidelinks,colorlinks=false,citecolor=black,linkcolor=black]{hyperref}
\usepackage{mathtools}
\usepackage{subcaption}
\usepackage{amssymb,amsmath,amstext,amsthm}
\usepackage{graphicx}
\usepackage{epstopdf}
\usepackage{color}
\usepackage{bm}
\usepackage{appendix}
\usepackage[T1]{fontenc}
\usepackage{bbm}
\usepackage{latexsym}

\usepackage{algorithm,algorithmic}

\usepackage{float}
\usepackage{verbatim}
\usepackage{lipsum}
\usepackage{braket}
\usepackage{ulem}

\newtheorem{theorem}{Theorem}
\newtheorem{lemma}{Lemma}

\newtheorem{definition}{Definition}

\begin{document}

\title{Learning linear optical circuits with coherent states}
\author{T.J.\,Volkoff}
\affiliation{Theoretical Division, Los Alamos National Laboratory, Los Alamos, NM, USA.}
\affiliation{Quantum Science Center}
\author{Andrew T. Sornborger}
\affiliation{CCS-3, Los Alamos National Laboratory, Los Alamos, NM, USA.}
\affiliation{Quantum Science Center}

\begin{abstract}
We analyze the energy and training data requirements for supervised learning of an $M$-mode linear optical circuit by minimizing an empirical risk defined solely from the action of the circuit on coherent states. When the linear optical circuit acts non-trivially only on $k<M$ unknown modes (i.e., a linear optical $k$-junta), we provide an energy-efficient, adaptive algorithm that identifies the junta set and learns the circuit. We compare two schemes for allocating a total energy, $E$, to the learning algorithm. In the first scheme, each of the $T$ random training coherent states has energy $E/T$. In the second scheme, a single random $MT$-mode coherent state with energy $E$ is partitioned into $T$ training coherent states. The latter scheme exhibits a polynomial advantage in training data size sufficient for convergence of the empirical risk to the full risk due to concentration of measure on the $(2MT-1)$-sphere. Specifically, generalization bounds for both schemes are proven, which indicate the sufficiency of $O(E^{1/2}M)$ training states ($O(E^{1/3}M^{1/3})$ training states) in the first (second) scheme.
\end{abstract}
\maketitle

\section{Introduction}\label{sec:intro}
 Spatially- or temporally-multiplexed linear optical circuits are the central components of devices that process photon-encoded quantum information and therefore play a vital role in photonic information processing \cite{TAN2019100030,eng}. 
 After the initial generation of nonclassical resource states using optical nonlinearities or heralded entanglement, a set of linear optical layers are responsible for constructing the complexity required for production of photonic or continuous-variable (CV) cluster states, which are the central resources for optical implementations of measurement-based quantum computation. 
 
 As an example of this structure, consider the Borealis device \cite{Madsen2022}, in which 216 pulses (length $\tau$) of single temporal mode CV squeezed states are sequentially routed through time-delay loops of length $\tau$, $6\tau$, and $36\tau$. Three variable beamsplitters, situated where the respective time-delay loops meet the mainline, allow the application of a linear optical circuit, generating long-range entanglement between the temporal modes. Validation that this device indeed produces large scale nonclassical CV states with high modal photon occupation was key for demonstrating  quantum supremacy for Gaussian boson sampling protocols implemented via the coupling of a photon number resolving detector to the Borealis linear optical circuit \cite{Madsen2022}. 
 
 Additionally, recent demonstration of
 scalable, programmable, time-multiplexed, and universal 3-mode linear optical operations \cite{takeda2} provides evidence that linear optical modules can be successfully implemented in near-term photonic quantum neural networks.

 Analogous to the problem of learning quantum $k$-juntas which act nontrivially on a $k$-qubit subregister \cite{yuen2}, in this work, we define the problem of learning linear optical $k$-juntas ($k$-LOJ), i.e., learning the modes on which a linear optical unitary acts nontrivially (the junta) and, potentially simultaneously, learning the action of the unitary on the junta modes. This problem is a specific example of methodology to characterize or certify the dynamical behavior of an optical circuit. Our approach to this problem consists of a hybrid quantum-classical algorithm utilizing coherent state inputs and CV SWAP tests \cite{vs} to compute and minimize an empirical risk function. 
 
 From a broader perspective, because each matrix element of a $k$-LOJ corresponds to a tunable physical parameter, statistical learning of $k$-LOJs as considered in the present work can be considered as a type of CV quantum process tomography, i.e., the task of obtaining classical estimates of CV quantum channels (including CV quantum state tomography) \cite{lvov,Rahimi-Keshari_2011,PhysRevA.92.022101,PhysRevLett.120.090501,obrien,Teo_2021,PhysRevA.102.012616}, for which specific methods have been introduced in the case of characterizing linear optical circuits \cite{PhysRevA.98.052327,rb13,Dhand_2016,laingobrien,spagnolo2,Anis_2012,Fedorov_2015,Katamadze_2021,Tillmann_2016,poot,PhysRevA.101.043809}. These methods have not previously made use of hybrid quantum-classical algorithms.
 
 We emphasize that  learning the parameters of more general CV quantum neural networks \cite{PhysRevResearch.1.033063} with the aim of compiling a target unitary dynamics (i.e., CV quantum compiling \cite{PRXQuantum.2.040327}) is generally less resource intensive than full quantum process tomography of the target unitary. This is because CV quantum compiling does not involve full classical characterization of all matrix elements of a target CV quantum circuit, but only learning an at-most-polynomial number of parameters that specify the ansatz CV quantum neural network. These parameters often take the form of the discrete structure specifications and interaction strengths.
 
 Our use of coherent states as training data is motivated by the fact that linear optical unitaries have a transitive action on the set of isoenergetic coherent states. Further, information obtained from measurements in the overcomplete coherent state or generalized coherent state basis has been useful for CV quantum state tomography \cite{PhysRevA.94.052327}. 
 
Statistical learning methods have been employed for classification of corrupted CV states from their approximate coherent state support (discretized Husimi $q$-function) \cite{PhysRevResearch.3.033278}. However, our empirical risk minimization uses neither entangled states nor CV state tomography subroutines. It relies solely on distinguishing coherent states via the CV SWAP test. Because linear optical unitaries preserve total photon number, the precision of the CV SWAP test depends only on the energy of the training coherent states. We define two empirical risk minimization (ERM) settings, \textbf{ERM1} and \textbf{ERM2}, which differ in the way that the training coherent states are generated. An energetic modification to \textbf{ERM1}, called \textbf{ERM1'}, allows quantitative comparison of the two training state generation methods.  
 
 The respective generalization performances of the ERM settings are found to depend on the number of modes of the linear optical circuit and the energy (in addition to the training set size, as usual). In the ERM settings that we consider, the respective covering number generalization bounds are proven using McDiarmid's inequality (\textbf{ERM1}, \textbf{ERM1'}) and L\'{e}vy's lemma (\textbf{ERM2}). These bounds may find other applications in quantum machine learning utilizing cost functions defined on products of many spheres, or spheres of large dimension. The problem of using nonclassical states of light for learning linear optics in the same hybrid quantum-classical framework (or otherwise) is left to future work.

 \begin{figure*}
    \centering
    \includegraphics[scale=1]{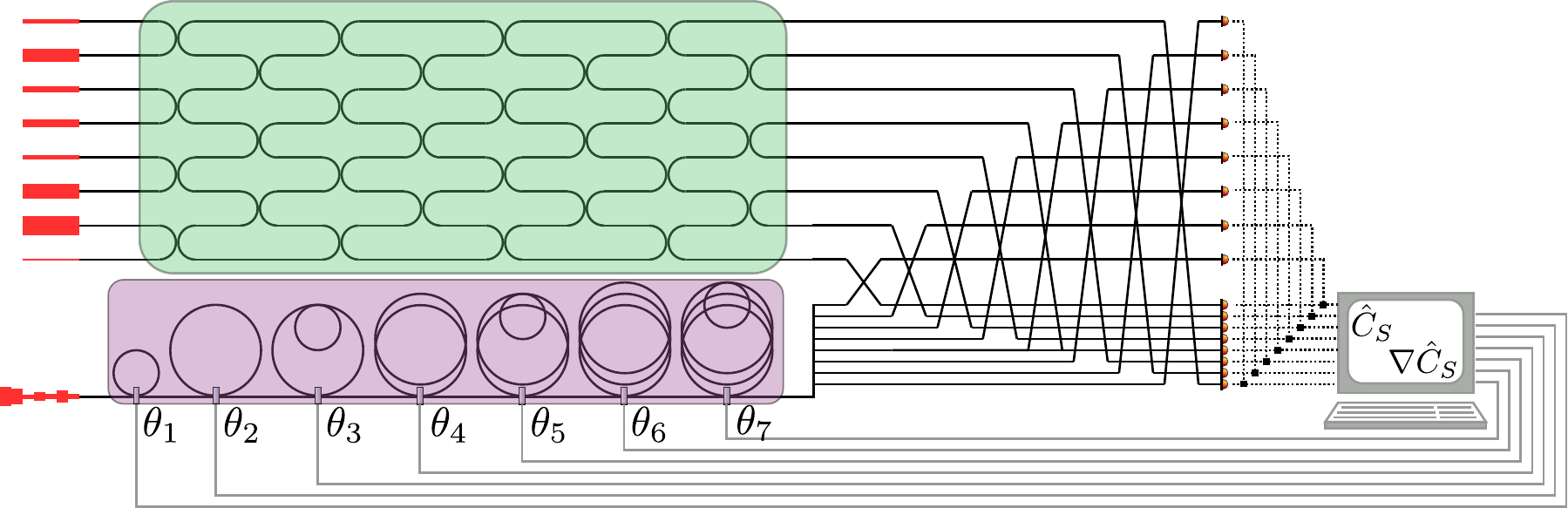}
    \caption{A method for $T=1$ empirical risk function calculation by CV SWAP test and classical optimization algorithm. A random training coherent state is injected into spatial modes of the target linear optical circuit $U$ (green) and temporal modes of the time-delay multiplexed circuit $V(\theta)$ (pink). The rectangular decomposition is used to signify $U$. The $\tau$-loop and $2\tau$-loop time-multiplexed circuit with programmable beamsplitters allows to implement arbitrary linear optical ansatz circuit $V(\theta)$ on $M=8$ temporal modes. The $8\tau$ signal from $V(\theta)$ is spatially demultiplexed and the CV SWAP test (beamsplitter network and photon counting) is performed.}
    \label{fig:schem}
\end{figure*}

We work with a common notational framework for CV systems. The vector of canonical operators (quadratures) is $R=(q_{1},\ldots,q_{M},p_{1},\ldots, p_{M})^{T}$, where the components satisfy the Heisenberg canonical commutation relation $[R_{i},R_{j}]=i\Omega_{i,j}$, and \begin{equation}
    \Omega=\begin{pmatrix}0&I\\-I&0\end{pmatrix}
\end{equation}
is the symplectic form on $\mathbb{R}^{2M}$. The photon number operator is $N_{M}:={1\over 2}R^{T}R$ with unique lowest eigenvector $\ket{0}^{\otimes M}$. A Heisenberg-Weyl coherent state is defined as $\ket{x}=D(x)\ket{0}^{\otimes M}$, where $D(x)=e^{-ix^{T}\Omega R}$, and the parameter $x$ has the physical meaning of the mean vector of the coherent state, viz., $\langle R \rangle_{\ket{x}} = x$.  

There is a one-to-one correspondence between a Gaussian unitary operation $U$ and a symplectic matrix  $T_{U}\in  Sp(2M,\mathbb{R})$ which acts on the quadrature vector operator as
\begin{equation}
    U^{\dagger}RU=T_{U}R.
\end{equation}
The Euler decomposition for $Sp(2M,\mathbb{R})$ implies that a Gaussian circuit can always be written as a composition of squeezing operations followed by a linear optical unitary \cite{Arvind,serafini} and, indeed, this is how NISQ circuits for Gaussian boson sampling are organized \cite{x8}. If $U$ is a linear optical unitary, then $T_{U}\in O(2M)\cap Sp(2M,\mathbb{R})=K(2M)$, the maximal compact subgroup of $Sp(2M,\mathbb{R})$.

\section{Linear optical $k$-juntas}
The maximal compact subgroup $K(2M)$ is seen to be isomorphic to the unitary group $U(M)$ by noting that if $G=\text{Re}G+i\text{Im}G$, then \begin{equation}\begin{pmatrix}
    \text{Re}G&\text{Im}G\\-\text{Im}G & \text{Re}G
\end{pmatrix}\in K(2M)
\end{equation}
and, conversely, that the $(1,1)$ and $(1,2)$ $M\times M$ blocks of any element of $K(2M)$ define the real and imaginary parts of an $M\times M$ unitary. However, due to connectivity and transmissivity restrictions in linear optical circuits,  linear optical unitaries often do not change every element of a coherent state mean vector.
\begin{definition}
An orthogonal matrix in $O\in O(2M)\cap Sp(2M,\mathbb{R})$ is a linear optical $k$-junta ($k$-LOJ) if there exists a subspace $\mathcal{K}\subset \mathbb{R}^{2M}$ of dimension $2M-2k$ and $T \in O(2k)\cap Sp(2k,\mathbb{R})$ such that $O=T\oplus I_{2M-2k}$. \label{def:oo}\end{definition}
This notion is straightforwardly related to the notion of quantum $k$-junta for linear optical unitaries \cite{wang,yuen2}.
\begin{lemma}
If $O$ is a $k$-LOJ and $U^{\dagger}RU=OR$, then there exists a set of $M$ CV modes for which $U$ is a quantum $k$-junta.
\end{lemma}
\begin{proof}
There is an isometric isomorphism from $\mathbb{R}^{2M}$ to itself that takes the subspace $\mathcal{K}$ in Definition \ref{def:oo} to the subspace spanned by $\lbrace e_{1},\ldots, e_{2M-2k}\rbrace$. Associate these $2M-2k$ orthogonal directions in $\mathbb{R}^{2M}$ with a vector of canonical operators $\tilde{R}_{\mathcal{K}}=(q_{1},p_{1},\ldots,q_{M-k},p_{M-k})^{T}$, and complete this vector to a vector of canonical operators for an $M$-mode CV system $\tilde{R}= \tilde{R}_{\mathcal{K}^{c}}\oplus \tilde{R}_{\mathcal{K}}$. Then the unitary $U$ associated with $O$ acts as identity on $\tilde{R}_{\mathcal{K}}$, i.e., $U=\mathbb{I}_{\mathcal{K}}\otimes U'$ for unitary $U'$ on $\mathcal{H}^{\otimes k}$.
\end{proof}

In practice, a circuit-based CV quantum information processing protocol is defined with respect to a predefined set of $M$ modes. A linear optical unitary that is a quantum $k$-junta with respect to these modes is then associated with $k$-LOJ in which the subspace $\mathcal{K}$ is defined by a $(M-k)$-subset of the modes. We now fix the set of $M$ CV modes and restrict our consideration to a $k$-LOJ with $\mathcal{K}$ having an orthonormal basis indexed by a subset of $[M]$.

\section{Coherent state ERM}
Before detailing a variational ERM algorithm for learning $k$-LOJ, we discuss features of the ERM problem for general linear optical unitaries. Consider a set of  training coherent states $S=\lbrace \ket{x^{(j)}}\rbrace_{j=1}^{T}$, $x^{(j)}\in \mathbb{R}^{2M}$, and empirical risk function for learning a target linear optical unitary $U$
\begin{align}
    \hat{C}_{S}(V)&={1\over 4T}\sum_{j=1}^{T}\Vert (\mathcal{U}-\mathcal{V})(\ket{x^{(j)}}\bra{x^{(j)}}) \Vert_{1}^{2} 
    \label{eqn:er}
\end{align}
where $V=V(\theta)$ is a linear optical unitary parameterized by $\theta$ and we denote a unitary channel by $\mathcal{U}(\cdot):=U(\cdot)U^{\dagger}$ for any unitary $U$. The empirical risks can be computed by using parity data from photon number counting measurement on corresponding pairs of modes of $U\ket{x}$ and $V\ket{x}$ as in the CV SWAP test, which has known energy cost when coherent state inputs are used \cite{vs}. Minimization of (\ref{eqn:er}) in the case of a single training coherent state is drawn in Fig.~\ref{fig:schem}, in which a target unitary (shown in a generic rectangular decomposition) is learned by updating the variable beamsplitters of a time-delay-multiplexed circuit. It is possible to use commonly implemented circuit architectures for $V(\theta)$, such as triangular decomposition or rectangular decomposition, but in the numerical optimizations in this work, we directly optimize over matrix elements so as to circumvent parametrization-dependent challenges in optimization. Unlike the case for unitary matrices in the computational basis on discrete-variable quantum systems, the matrix elements of $V$ are directly tunable because they are simply the couplings between field quadratures. 

We will consider two ways of enforcing an energy constraint on the training set $S$. In \textbf{ERM1}, $\Vert x^{(j)}\Vert =\sqrt{2E}$ for all $j$, so each training coherent state has energy $E$. In \textbf{ERM2}, one considers $T$ projections of an $MT$-mode coherent state $\ket{x}$ to the respective $M$-mode subsystems using the projections $P_{j}:\mathbb{R}^{2MT}\rightarrow \mathbb{R}^{2M}$ with $P_{j}x=(x_{2M(j-1)+1},\ldots,x_{2Mj})$, $j=1,\ldots, T$. One defines the local $M$-mode coherent states by $\ket{x^{(j)}}=\ket{P_{j}x}$. Taking $\Vert x\Vert = \sqrt{2E}$, the training set $S$ has total energy $E$, regardless of its size. Note that for $T>1$, \textbf{ERM1} has strictly greater total energy than \textbf{ERM2}; later we will introduce a modification to \textbf{ERM1}, called \textbf{ERM1'}, in which the training data set has total energy equal to that of \textbf{ERM2}. These  energy-constrained empirical risk functions correspond to empirical risk minimization tasks with serial and parallel training coherent state sources, respectively.
When variationally learning a linear optical unitary, the empirical risk serves as an estimate of the full risk
\begin{align}
    C_{\textbf{ERM1}}(V)&:= {1\over 4}\int_{S^{2M-1}} \Vert (\mathcal{U}-\mathcal{V})(\ket{x}\bra{x})\Vert_{1}^{2}\nonumber \\
            C_{\textbf{ERM2}}(V)&:= {1\over 4}\int_{S^{2MT-1}}\Vert (\mathcal{U}-\mathcal{V})(\ket{P_{1}x}\bra{P_{1}x})\Vert_{1}^{2} 
    \label{eqn:risk}
\end{align}
where both integrals are over the uniform measure on the respective sphere (of radius $\sqrt{2E}$ for \textbf{ERM1} and \textbf{ERM2}). Note that $C_{\textbf{ERM2}}(V)$ depends on $T$ because the local training coherent state $\ket{P_{j}x}$ in the $j$-th subsystem is obtained from a $MT$-mode coherent state.  The following chain of inequalities shows that each term in the sum (\ref{eqn:er}) has a simple upper bound in terms of the orthogonal matrices associated to $U$ and $V$, and the energy constraint:
\begin{align}
    &{} {1\over 2}\Vert (\mathcal{U}-\mathcal{V})\ket{x}\bra{x}) \Vert_{1}  = \sqrt{1-\vert \langle x\vert U^{\dagger}V\vert x\rangle\vert^{2}} \nonumber \\
    &= \sqrt{1- e^{-{1\over 2}x^{T}(O_{U}-O_{V})^{T}(O_{U}-O_{V})x}}\nonumber \\
    &\le \sqrt{1-e^{-{1\over 2}\Vert x\Vert^{2}\Vert O_{U}-O_{V}\Vert^{2}}}\nonumber \\
    &\le \sqrt{E}\Vert O_{U}-O_{V}\Vert
    \label{eqn:ineq1}
\end{align}
with $\Vert x\Vert=\sqrt{2E}$.
 This fact implies Lipschitz continuity of the empirical risk and the full risk in both settings.
\begin{lemma}
    Let $W$ and $V$ be linear optical unitaries such that $\Vert O_{W}-O_{V}\Vert \le {\epsilon \over  \sqrt{E}}$. Then,
    \begin{align}
        &{} \vert \hat{C}_{S}(W)-\hat{C}_{S}(V) \vert \le \epsilon \text{ with probability } 1\nonumber \\
        &{} \vert C_{\textbf{ERM1}}(W)-C_{\textbf{ERM1}}(V)\vert \le \epsilon.
    \end{align}
    If $\Vert O_{W}-O_{V}\Vert \le \epsilon \sqrt{2MT-1\over E(2M+1)}$, then
        \begin{align}
        &{}\vert C_{\textbf{ERM2}}(W)-C_{\textbf{ERM2}}(V)\vert \le \epsilon.
    \end{align}
    \label{lem:ooo}
    \end{lemma}
    \begin{proof}
        We show the first inequality; the proof of the second is the same, replacing the discrete measure over training states with the uniform measure.
        \begin{align}
       &{} \vert \hat{C}_{S}(W)-\hat{C}_{S}(V) \vert  \nonumber \\
       &\le {1\over T}\sum_{j=1}^{T}\Big\vert  \vert \langle x^{(j)}\vert UW^{\dagger}\vert x^{(j)}\rangle\vert^{2}  - \vert \langle x^{(j)}\vert UV^{\dagger}\vert x^{(j)}\rangle\vert^{2}\Big\vert \nonumber \\
       &\le
         {1\over 2}\max_{j}\Vert W\ket{x^{(j)}}\bra{x^{(j)}}W^{\dagger}-V\ket{x^{(j)}}\bra{x^{(j)}}V^{\dagger}\Vert_{1}  \nonumber \\
        &\le \epsilon
    \end{align}
    where we note that for \textbf{ERM2}, $\Vert x^{(j)}\Vert \le \sqrt{2E}$ trivially because $\Vert x\Vert =\sqrt{2E}$.
    The first inequality uses the definition of trace distance for pure states and the triangle inequality, the second inequality is standard (Eq.~(9.96) of \cite{wilde}), the third inequality is (\ref{eqn:ineq1}).

    For \textbf{ERM2} the Lipschitz constant of the cost function is smaller. Note that 
    \begin{align}
        &{} \vert C_{\textbf{ERM2}}(W)-C_{\textbf{ERM2}}(V)\vert \le \nonumber \\
        &{} \int_{S^{2MT-1}} \vert \langle P_{1}x\vert UW^{\dagger}\vert P_{1}x\rangle \vert^{2} -\vert \langle P_{1}x\vert UV^{\dagger}\vert P_{1}x\rangle \vert^{2}\vert \nonumber \\
        &{} \le {1\over \sqrt{2}}\Vert O_{W}-O_{V}\Vert \int_{S^{2MT-1}} \Vert P_{1}x\Vert.
    \end{align}
    Carrying out the integral using the marginal in (\ref{eqn:marg}) and Gautschi's inequality gives that the last line is upper bounded by $\Vert O_{W}-O_{V}\Vert \sqrt{E (2M+1)/(2MT-1)}$.
    \end{proof}

    We note that the integrals defining the full risks (\ref{eqn:risk}) do not have closed form expressions, but can be approximated from series expansions, similar to hypergeometric functions of matrix argument \cite{koev}. For instance, $1-C_{\textbf{ERM1}}(V)$ is given by
    \begin{equation}
      2E^{2M-1}\sum_{i_{1},\ldots,i_{2M}=0}^{\infty}{\prod_{j=1}^{2M}{1\over i_{j}!}\left( -{\kappa_{j}^{2}\over 2}\right)^{i_{j}}\Gamma(i_{j}+{1\over 2})\over  \Gamma(M+\sum_{j=1}^{2M}i_{j})}  .
    \end{equation}
    where $\kappa_{1},\ldots,\kappa_{2M}$ are the singular values of $O_{U}-O_{V}$.

To determine whether numerical minimization of (\ref{eqn:er}) successfully learns $U$, in both Figs.~\ref{fig:f1}a and \ref{fig:f1}b we observe that below the critical value $T=M$, minimization of (\ref{eqn:er}) fails to return a $W\in K(2M)$ that approximates the target $O_{U}$ closely in Frobenius norm.
\begin{figure*}
    \centering
    \includegraphics[scale=0.545]{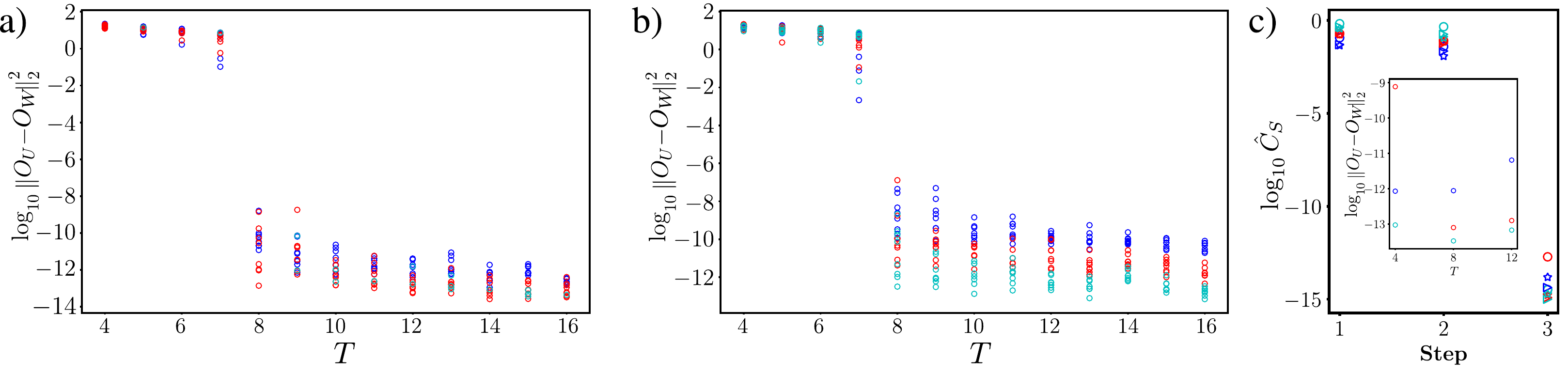}
    \caption{Empirical risk minimization for a uniform random $8$-mode target linear optical unitary, $U$, with 8 uniformly random coherent state training datasets of size $T\in \lbrace 4,\ldots,16\rbrace$. If a training dataset did not result in finding a $W$ such that $\hat{C}_{S}(W)<10^{-7}$, it was not shown.  Quality of the minimization output is quantified by $\Vert O_{U}-O_{W}\Vert_{2}^{2}$. a) (\textbf{ERM1}) Blue, $E=1$; Red, $E=4$; Cyan $E=16$. \textbf{ERM1} failed to obtain a minimum for $T < 8$. This is expected since the empirical risk is not faithful for a smaller training dataset. b) (\textbf{ERM2}) Blue, $E=1$; Red, $E=4$; Cyan $E=16$. The constraint $\Vert V^{\dagger}V-\mathbb{I}_{M}\Vert_{2}^{2}\le 10^{-6}$ was enforced on the variational matrix $V$ during this simulated ERM to ensure approximate unitarity. Note that many more cyan ($E = 16$) datapoints are found, indicating that \textbf{ERM2} has a more favorable landscape for optimization relative to \textbf{ERM1}. c) Log$_{10}$ of the \textbf{ERM2} empirical risk minimal value for learning a single, random 4-LOJ on $M=8$ modes using Algorithm \ref{alg:lo} ($J=\lbrace 3,4,5,8\rbrace$ for this target unitary). Circles are T=4, triangles are $T=8$, stars are $T=12$. Note that $T\ge 4$ is necessary for faithful empirical risk for a target $4$-LOJ. The junta size determines the number of steps that the algorithm takes; here, the junta size is $4$ and therefore $3$ steps of the algorithm suffices. Inset shows distance of output $O_{W}$ from $O_{U}$ after third step for various $T$. We see that the distance of the learned matrix to the target matrix is smaller for higher energy.}
    \label{fig:f1}
\end{figure*}
At $T= M$, we observe a transition to successful learning of $U$ (if success is characterized by an appropriate value of the Hilbert-Schmidt distance). This transition is explained by the fact that a $2M\times 2M$ symplectic orthogonal matrix is a function of two $M\times M$ block submatrices, so $M$ linearly independent training coherent states are sufficient to make $\text{C}_{\textbf{ERM1(2)}}$ a faithful cost function (almost everywhere with respect to the measure governing the training set $S$), i.e., the unique global minimum occurs for $V=U$. 

However, even when the empirical risk is almost always faithful (which is the case for $T\ge M$), the performance of the minimization algorithm depends on how the energy scales with the number of modes, due to the energy-dependent barren plateau phenomenon \cite{Volkoff2021}. Specifically, when the energy of a multimode coherent state scales at least linearly with the number of modes, the expected magnitude of the gradient of the empirical risk exponentially decreases to zero with respect to increasing $M$, rendering gradient-descent-based minimization algorithms ineffective. This energy-dependent barren plateau phenomenon is the cause of the reduced effectiveness of the learning algorithm \textbf{ERM1} for large energies. For example, in Fig. \ref{fig:f1}a, the minimization is successful for $E=4$, but the minimization is rarely successful for $E=16$.  

To see how this problem is lessened for $\textbf{ERM2}$, we appeal to the marginal probability density of each of the $T$ $M$-mode coherent states comprising a training set in \textbf{ERM2}. Without loss of generality \cite{muirhead}, 
\begin{equation}
    p(x^{(1)})\propto {\theta(2E - \Vert x^{(1)}\Vert^{2})\over (2E)^{MT-1}}(2E-\Vert x^{(1)}\Vert^{2})^{M(T-1)-1}.
    \label{eqn:marg}
\end{equation}
From this marginal density, one finds that the expected energy of each of the individual training coherent states $\ket{x^{(i)}}$ in \textbf{ERM2} is $E/T$. If the energy $E$ is chosen to scale linearly with $M$, then increasing $T$ reduces the impact of the exponential factor that causes the energy-dependent barren plateau phenomenon, which eases optimization \cite{Volkoff2021}. In fact, the data in Fig.~\ref{fig:f1}b is closely replicated by the scheme of \textbf{ERM1} modified by requiring that each training coherent state have energy $E/T$, i.e., $\Vert x^{(i)}\Vert=\sqrt{2E/T}$. We will call this modified scheme \textbf{ERM1'}. However, in the next section, we see that concentration of $C_{\textbf{ERM2}}$ on the sphere embedded in $\mathbb{R}^{2MT}$ causes an increased rate of convergence to the associated risk function (with respect to $T$) compared to $C_{\textbf{ERM1'}}$.

Before carrying out an analysis of generalization bounds for the empirical risk minimization settings, we first introduce an energy-efficient method for learning a $k$-LOJ. Specifically, the target orthogonal matrix is promised to be a $k$-LOJ, but the junta set $J\subset [M]$, $\vert J\vert=k$ is unknown. The energy-efficient adaptive method in Algorithm \ref{alg:lo} is based on the fact that if the subset $J$ were perfectly known, then it would be wasteful to use training coherent states that are non-vacuum on the $J^{c}$ modes. The eventual success of Algorithm 1 is predicated on the success of the empirical risk minimization when $T$ is at least the size of the junta set, which can be demonstrated utilizing an analysis similar to that discussed above for Figs.~\ref{fig:f1}a, b.

Algorithm \ref{alg:lo} finds the junta set $J$ and learns the target $k$-LOJ using $O((M+k)^{2})$ minimizations of a faithful empirical risk, e.g., $\hat{C}_{S}$ defined by a training set of size at least $k$. 
\begin{figure}
\begin{algorithm}[H]
  \caption{$k$-LOJ learning with adaptive linear optical ansatze}
  \label{EPSA}
   \begin{algorithmic}[1]
   \STATE Given: $k$-LOJ $O_{U}$ with unknown $k$ and junta set $J$
   \STATE $J=\emptyset$
   \STATE $\mathcal{J}_{2}=\lbrace \lbrace i,j\rbrace\subset [M]:i<j\rbrace$
    \STATE Ansatz $V^{\lbrace i,j\rbrace}(\alpha)$ acts nontrivially on modes $\lbrace i,j\rbrace$. Compute $\min_{\alpha}\hat{C}_{S}(V^{\lbrace i,j\rbrace}(\alpha))$ for all $\lbrace i,j\rbrace\in \mathcal{J}_{2}$ and return the minimal value $c_{2}$, which occurs for mode pairs $\lbrace i^{(1)},j^{(1)}\rbrace$, $\lbrace i^{(2)},j^{(2)}\rbrace$, $\ldots$ 
    \STATE $J\leftarrow \cup_{\ell}\lbrace i^{(\ell)},j^{(\ell)}\rbrace$
    \IF{$c_{2}<10^{-10}$}
    \STATE Return $J$ and $\bigotimes_{\ell}V^{\lbrace i^{(\ell)},j^{(\ell)}\rbrace}(\alpha^{\lbrace i^{(\ell)},j^{(\ell)}\rbrace}_{c})$
    \ELSIF{$c_{2}\ge 10^{-10}$}
    \STATE $\mathcal{J}_{3}=\lbrace J\cup \lbrace l\rbrace\subset [M]:l\notin J\rbrace$
    \STATE Ansatz $V^{J\cup \lbrace l\rbrace}(\beta)$ acts nontrivially on modes $J\cup \lbrace l\rbrace$. Compute $\min_{\beta}\hat{C}_{S}(V^{J\cup \lbrace l\rbrace}(\beta))$ for all $J\cup \lbrace l\rbrace\in\mathcal{J}_{3}$ and return the minimal value $c_{3}$, which occurs for mode sets $J\cup \lbrace l^{(1)}\rbrace$, $J\cup \lbrace l^{(2)}\rbrace$\, $\ldots$ 
    \ENDIF
    \STATE $J\leftarrow J \cup_{\ell}\lbrace i^{(\ell)}\rbrace$
        \IF{$c_{3}<10^{-10}$}
    \STATE Return $J$ and $\bigotimes_{\ell}V^{J\cup \lbrace l^{(\ell)}\rbrace }(\beta^{J\cup \lbrace l^{(\ell)}\rbrace}_{c})$
    \ELSIF{$c_{3}\ge 10^{-10}$}
    \STATE $\mathcal{J}_{4}=\lbrace J\cup \lbrace l\rbrace\subset [M]:l\notin J\rbrace$
    \STATE Ansatz $V^{J\cup \lbrace l\rbrace}(\gamma)$ acts nontrivially on modes $J\cup \lbrace l\rbrace$. Compute $\min_{\gamma}\hat{C}_{S}(V^{J\cup \lbrace l\rbrace}(\gamma))$ for all $J\cup \lbrace l\rbrace\in\mathcal{J}_{4}$ and return the minimal value $c_{4}$, which occurs for mode sets $J\cup \lbrace l^{(1)}\rbrace$, $J\cup \lbrace l^{(2)}\rbrace$, $\ldots$
    \ENDIF
    \STATE $J\leftarrow J \cup_{\ell}\lbrace i^{(\ell)}\rbrace$
    \STATE \textbf{if} $c_{4}<10^{-10}$
    \STATE \hspace{0.4cm} Return $J$ and $\bigotimes_{\ell}V^{J\cup \lbrace l^{(\ell)}\rbrace }(\gamma^{J\cup \lbrace l^{(\ell)}\rbrace}_{c})$
    \STATE Etc.
   \end{algorithmic}
   \label{alg:lo}
\end{algorithm}
\end{figure}
In the pseudocode for Algorithm \ref{alg:lo}, we have suppressed the subroutine generating the training states for empirical risk minimization with the $2$-LOJ ansatze defined by members of $\mathcal{J}_{2}$, the $3$-LOJ ansatze defined by members of $\mathcal{J}_{3}$, etc. To define this subroutine, first note that at the iteration defined by $\mathcal{J}_{k}$, one should use training sets of size $T\ge k$ so that the empirical risk function is faithful for learning $k$-LOJ and, therefore, so that the algorithm can in principle terminate. Then, to utilize the $k$ training states in an energy-efficient way within \textbf{ERM2}, note that identifying a minimal value $c_{2}$ requires an energy at least $E_{2}$ for each of the trials in $\mathcal{J}_{2}$, which depends on the specific implementation (one can see this by noting that, e.g., $\hat{C}_{S}\rightarrow 1$ for $\Vert x\Vert \rightarrow 0$). Similarly, in a general iteration $j$, identifying the minimal value $c_{j}$ requires a minimal total energy $E_{j}$. Within the scheme \textbf{ERM2}, this procedure consumes energy at most ${M\choose 2}E_{2} + (M-2)E_{3}+\ldots + (M-k+1)E_{k}$. Numerical data obtained by applying Algorithm \ref{alg:lo} to a random  $M=8$ mode $4$-LOJ is shown in Fig. \ref{fig:f1}c).

Note that there are many non-generic instances of $U$ (but commonly encountered in actual devices)  where this algorithm would terminate in the first junta set update. For instance, consider $U$ being a tensor product of generalized beamsplitters on modes $(1,2)$, $(3,4)$, $(M,M-1)$. After the junta set update, one would conclude that not only is the junta set $J=[M]$, but the $U$ is given by $\bigotimes_{j=1}^{M-1}V^{(j,j+1)}(\alpha^{(i,j)}_{c})$ to high accuracy. In Algorithm \ref{alg:lo}, we start with pairs of modes because one can consider generalized beamsplitters as a fundamental linear optical element, which subsumes the case of a tensor product of single-mode phase shifts. Strictly speaking, the numerical precision of the optimization algorithm will allow to identify a mode set $J\cup \lbrace l\rbrace$ achieving the minimal value $c_{2}$, $c_{3}$, etc., so there should be a small finite tolerance on what value is considered ``minimal'' so as not to do more ERM iterations than needed. It should be noted that the junta set updates in Algorithm \ref{alg:lo} do not make use of quantum coherence. Introducing a discrete variable quantum register which controls the junta set updating procedure and carrying out, e.g., Grover's algorithm for unstructured search, on that register could result in learning $U$ after sublinear (in $M$) ERM  runs.

Finally, we note that there are alternative algorithms that allow to identify the junta set $J$, but without learning the target $k$-LOJ. In Ref.~\cite{yuen2}, the junta set of a multiqubit unitary is determined with high probability by sampling its Pauli spectrum using the Choi state of the unitary. Since CV systems do not have an analog of the Pauli spectrum, a different quantum algorithm is required. In the case of finding the junta set for $k$-LOJ, the fact that linear optical unitaries preserve the set of coherent states allows a semiclassical algorithm utilizing only pure coherent states. We assume preparation of the $4M^{2}$-mode coherent state
\begin{equation}
    \left( U\ket{x}_{A}\otimes \ket{x}_{B}\right)^{\otimes M}\otimes \left( U\ket{y}_{C}\otimes \ket{y}_{D}\right)^{\otimes M}
\end{equation}
where $A=(A_{1},\ldots,A_{M})$ is an $M$-mode CV system (similarly for $B$, $C$, $D$), and $P_{j}x\neq \lambda P_{j}y$ for any $\lambda\in \mathbb{R}$, where $P_{j}$ is the projector to the $j$-th component of a vector in $\mathbb{R}^{2M}$. The above state is equal to
\begin{align}
    &{} \bigotimes_{j=1}^{M}\left[ \ket{P_{j}O_{U}x}\bra{P_{j}O_{U}x}_{A_{j}} \otimes \ket{P_{j}x}\bra{P_{j}x}_{B_{j}} \right. \nonumber \\
    &{}\left. \otimes \ket{P_{j}O_{U}y}\bra{P_{j}O_{U}y}_{C_{j}} \otimes \ket{P_{j}y}\bra{P_{j}y}_{D_{j}}\right].\label{eqn:fgfg}
\end{align}
For each $j\subset [M]$, the CV SWAP test is carried out on (\ref{eqn:fgfg}) between modes $A_{j}$ and $B_{j}$ and between modes $C_{j}$ and $D_{j}$ \cite{vs}. If both succeed, the mode $j$ is appended to the complement of the junta set. This procedure is repeated according to the energy-dependent complexity results of \cite{vs}, producing the junta set to a desired accuracy. 

\section{Generalization bounds}

In quantum machine learning with few training data, and, specifically in the task of learning quantum dynamics, the empirical risk often does not define a faithful cost function. Generalization bounds allow to bound the probability that the value of the empirical risk differs by some small amount from the value of the full risk. Obtaining the minimal value of the full risk on a hypothesis unitary guarantees that the dynamics has been learned exactly.
When proving generalization bounds for full risk functions that are defined via expectation values of observables that are only constrained roughly (e.g., by their operator norm), the concentration inequality used to derive the generalization bound must be correspondingly general. For instance, in \cite{Caro2022} the only quantity constraining the empirical risk is the norm of the observable defining the loss function, so the relevant concentration inequalities are for independent, identically distributed random variables in an interval defined by that norm. Since the empirical risk functions in this paper are functions of high-dimensional sphere-valued random variables with the sphere radius determined by the energy and risk minimization setting, concentration on the sphere is the relevant concept for determining the convergence rate of the empirical risk to the full risk.

Let $X$ be a random variable taking values on the $D$-sphere in $\mathbb{R}^{D+1}$ with radius $R$, and let the distribution of $X$ have uniform density on the sphere. Consider the Gaussian function
\begin{equation}
f(X)=e^{-{1\over 2}X^{T}L(U,V)X}
\label{eqn:lip}
\end{equation}
where $L(U,V):=(O_{U}-O_{V})^{T}(O_{U}-O_{V})$ is a positive, symmetric matrix. The restriction of the function (\ref{eqn:lip}) to the $D$-sphere of radius $R$ is $\kappa$-Lipschitz where \begin{equation}\kappa\le \Vert L(U,V)\Vert Re^{-\lambda_{\text{min}}(L(U,V)) E} \le R\Vert L(U,V)\Vert .
\end{equation}
The concentration of a function, $f$, of a sphere-valued random variable about its expectation, $E(f)$, is related to the Lipschitz constant, $\kappa$, through  L\'{e}vy's lemma \cite{Hayden2006} 
\begin{equation}
P(\vert f(X)-E(f(X))\vert \ge \eta) \le 2e^{-C_{1}(D+1)\eta^{2}\over \kappa^{2}}.
\label{eqn:levy}
\end{equation}
The generalization bounds for learning a linear optical unitary that appear in Theorem \ref{thm:levthm} indicate a qualitative difference in the size of the coherent state training set required for convergence of \textbf{ERM1} and \textbf{ERM2} learning protocols.

\begin{theorem}
(Generalization bounds for learning  $O_{U}$) Let $V(\theta_{S})$ be the unitary minimizing (\ref{eqn:er}) for each training set $S$. Then, with $C_{1}=(9\pi^{3}\ln 2)^{-1}$, the following bound holds with probability at least $1-\delta$ over training sets $S$ of energy $E$ and size $T$
\begin{align}
    &{} \vert C_{\textbf{ERM2}}(V(\theta_{S}))-\hat{C}_{S}(V(\theta_{S}) )\vert \nonumber \\
    &< \sqrt{ {16EM \log 6\sqrt{C_{1}MT^{3}}  \over C_{1}T^{3}} + {16E\log{2\over \delta} \over C_{1}MT^{3}}} \nonumber \\
    &+ 2\sqrt{E\over C_{1}MT^{3}}.
    \label{eqn:gb1}
\end{align}
if $S$ satisfies the energy constraint in \textbf{ERM2}, and
\begin{align}
    &{} \vert C_{\textbf{ERM1}}(V(\theta_{S}))-\hat{C}_{S}(V(\theta_{S}) )\vert \nonumber \\
    &< \sqrt{ {32EM^{2} \log 6\sqrt{T}  \over T} + {32E\log{2\over \delta} \over T}} \nonumber \\
    &+ 2\sqrt{E\over T}
    \label{eqn:gb2}
\end{align}
if $S$ satisfies the energy constraint in \textbf{ERM1}.
\label{thm:levthm}
\end{theorem}
\begin{proof}
    We first prove (\ref{eqn:gb1}) because it makes use of L\'{e}vy's lemma. Let $\mathcal{N}_{{\epsilon/ \sqrt{E}}}=\lbrace g_{j}\rbrace_{j}$ be an $\epsilon/\sqrt{E}$ cover of $O(2M)$. Consider a fixed training set $S$. From Lemma \ref{lem:ooo}, it follows that there is a $g\in \mathcal{N}_{\epsilon/\sqrt{E}}$ such that (with probability 1 over $S$)
    \begin{align}
        \big\vert \vert \hat{C}_{S}(g)-C(g)\vert - \vert \hat{C}_{S}(V(\theta_{S}))-C(V(\theta_{S}))\vert \big\vert \le 2\epsilon.
        \label{eqn:bbg}
    \end{align}
    From this, it follows that
    \begin{align}
        &{}P_{S}\left[  \vert \hat{C}_{S}(V(\theta_{S}))-C(V(\theta_{S}))\vert \ge \eta + 2\epsilon \right] \nonumber \\
        &\le P_{S}\left[  \exists g \in \mathcal{N}_{\epsilon/\sqrt{E}}: \, \vert \hat{C}_{S}(g)-C(g)\vert \ge \eta  \right]
        \label{eqn:ppi}
    \end{align}
    because if the event on the left hand side holds and $\vert \hat{C}_{S}(g)-C(g)\vert< \eta$ for all $g\in \mathcal{N}_{\epsilon/\sqrt{E}}$, then (\ref{eqn:bbg}) implies that $V(\theta_{S})$ cannot be in an $\epsilon/\sqrt{E}$ neighborhood of any $g$, which is a contradiction. For \textbf{ERM2}, the Lipschitz constant in Levy's lemma (\ref{eqn:levy}) is upper bounded using
    \begin{equation}
        \max_{\substack{ x\in \mathbb{R}^{2MT}\\ \Vert x \Vert=\sqrt{2E}}}\Vert \nabla \hat{C}_{S}(g)\Vert \le {\sqrt{2E}\over T}\Vert L(U,g)\Vert \le {4\sqrt{2E}\over T},
    \end{equation}
    where $L(U,g)$ is defined below (\ref{eqn:lip}).
    Applying the union bound gives that the right hand side of (\ref{eqn:ppi}) is  
    \begin{align}
        &\le 2\vert \mathcal{N}_{{\epsilon\over \sqrt{E}}}\vert e^{-C_{1}MT^{3}\eta^{2}\over 16E}\nonumber \\
        &\le 2\left( {6\sqrt{E}\over \epsilon}\right)^{M^{2}}e^{-C_{1}MT^{3}\eta^{2}\over 16E}.
    \end{align}
    Taking
    \begin{equation}
        \eta =\sqrt{{16EM\log {6\sqrt{E}\over \epsilon} \over C_{1}T} + {16E \log{2\over \delta}\over C_{1}MT}}
    \end{equation}
    gives that the right hand side of (\ref{eqn:ppi}) is upper bounded by $\delta$. Taking $\epsilon = \sqrt{E\over C_{1}MT}$ so that $\eta$ and $2\epsilon$ have the same scaling with $T$, it follows that the event (\ref{eqn:gb1}) holds with probability greater than $1-\delta$.

    To prove (\ref{eqn:gb2}) one cannot use L\'{e}vy's lemma because $\hat{C}_{S}$ is defined on a product of spheres $(S^{2M-1})^{\times T}$ instead of a single sphere. However, McDiarmid's inequality \cite{vershynin} is suitable for analyzing concentration of functions of i.i.d. random variables when an upper bound on the Lipschitz constant of the function is available. Under the assumptions of \textbf{ERM1}, one notes that
    \begin{equation}
        \max_{\substack{ (x^{(1)},\ldots,x^{(T)})\\ \Vert x^{(1)} \Vert =\cdots=\Vert x^{(T)} \Vert=\sqrt{2E}} }\Vert \nabla \hat{C}_{S}(g)\Vert \le 4\sqrt{2E}.
        \label{eqn:liperm1}
    \end{equation}
    Therefore, due to fact that $\hat{C}_{S}$ is a symmetric function on $(S^{2M-1}(\sqrt{2E}))^{\times T}$, the difference between $\hat{C}_{S}$ evaluated at $(x^{(1)},\ldots,x^{(j)},\ldots,x^{(T)})$ and $\hat{C}_{S}$ evaluated at $(v^{(1)},\ldots,y^{(j)},\ldots,x^{(T)})$ for any $j\in [T]$ is upper bounded by
    \begin{equation}
        {\Vert L(U,g)\Vert \sqrt{2E}\over T}\Vert x^{(j)}-y^{(j)}\Vert \le {4E\Vert L(U,g)\Vert \over T}.
    \end{equation}
    Applying McDiarmid's inequality and the union bound gives that the right hand side of (\ref{eqn:ppi})  is now upper bounded by
    \begin{equation}
        2\left({6\sqrt{E}\over \epsilon}\right)^{M^{2}}e^{-{\eta^{2}T\over 32 E}}.
    \end{equation}
   Taking 
   \begin{equation}
       \eta^{2}={32EM^{2}\log {6\sqrt{E}\over \epsilon}\over T} + {32E\log{2\over \delta}\over T}
   \end{equation}
   to get the $1-\delta$ guarantee, and taking $\epsilon={\sqrt{E\over T}}$ gives the stated generalization bound.
    
\end{proof}
To prove a generalization bound in the modified scheme \textbf{ERM1'}, which has the same total energy of the training set $S$ as \textbf{ERM2}, one should follow the proof of (\ref{eqn:gb2}) using McDiarmid's inequality, but the right hand side of (\ref{eqn:liperm1}) should be replaced by $4\sqrt{2E\over T}$ owing to the $\Vert x^{(i)}\Vert=\sqrt{2E/T}$ constraint. One obtains that $T=O(M\sqrt{E})$ training states are sufficient for generalization. 

\section{Discussion}
The importance of linear optical circuits for photonic quantum information processing is indicated by the fact that techniques for automated quantum experiment design use linear optical circuits as primary examples and testbeds \cite{zeil,PhysRevX.11.031044,PhysRevLett.124.010501}. In the present work, empirical risks for learning linear optical $k$-juntas are defined and calculated using coherent state training datasets with a fixed distribution of energy. Analysis of empirical risks defined by general energy-constrained CV Gaussian training states would not only allow optically non-classical training states, but could also extend the learning setting to include an entangled memory register \cite{PhysRevLett.128.070501,PRXQuantum.2.040327}. Improvements to sample complexity of learning linear optics are expected in this setting, following physically from the optimal sensitivity of non-classical CV Gaussian states to optical phase shifts under a mean energy constraint \cite{PhysRevA.90.025802}.

In terms of practical implementation, we consider \textbf{ERM2} to be well-suited to the scenario when the target linear optical circuit $U$ is accessible in parallel by distributing $T$ $M$-mode probe states. In constrast, \textbf{ERM1} or \textbf{ERM1'} seem more relevant if access to the unitary is scarce, e.g., in time. Regardless of the setting, we envision the training coherent states being generated by a circuit of random generalized beamsplitters applied to parent beam pulses. Lastly, we note that Algorithm 1 can be made both more energy and sample-efficient if a restriction to a specific hypothesis class of ansatze $V(\theta)$ is made (even probabilistically, as would be relevant if a sharper prior on the target linear optical unitary is known). In these cases, it is also possible that more refined learning algorithms for linear optical circuits \cite{Kuzmin:21,Spagnolo2017,Flamini2017} could be applied in lieu of Algorithm \ref{alg:lo}.

\acknowledgements
The authors thank Lukasz Cincio, Y. Suba\c{s}{\i}, and Z. Holmes for useful discussions. This work was supported by the Quantum Science Center (QSC), a National Quantum Information Science Research Center of the U.S. Department of Energy (DOE). Los Alamos National Laboratory is managed by Triad National Security, LLC, for the National Nuclear Security Administration of the U.S. Department of Energy under Contract No. 89233218CNA000001.
\bibliography{bibl.bib}
\bibliographystyle{plainnat}

% \onecolumngrid
% \appendix
% \section{\label{app:aa}...}

\end{document}